\documentclass[journal]{IEEEtran}
\usepackage{amssymb,slashbox,epstopdf,booktabs,paralist}
\usepackage[cmex10]{amsmath}
\usepackage{amsfonts}
\usepackage{url,epstopdf}
\usepackage{graphicx}
\usepackage[ruled]{algorithm}
\usepackage[noend]{algorithmic}

\newtheorem{theorem}{\textbf{Theorem}}

\ifCLASSINFOpdf
\else
\fi
\begin{document}
%
\title{Cooperative Secret Key Generation from Phase Estimation in Narrowband Fading Channels}

\author{Qian~Wang,~\IEEEmembership{Student~Member,~IEEE,}
        Kaihe~Xu,~\IEEEmembership{Student~Member,~IEEE,}
        and~Kui~Ren,~\IEEEmembership{Senior~Member,~IEEE} 
\IEEEcompsocitemizethanks{
\IEEEcompsocthanksitem The research of Kui Ren is partially supported by the US National
Science Foundation under grants CNS-0831963 and CNS-1117811.
\protect
\IEEEcompsocthanksitem Qian Wang, Kaihe Xu, and Kui Ren are with the Department
of Electrical and Computer Engineering, Illinois Institute of Technology, 3301 South Dearborn Street, Suite 103 Siegel Hall, Chicago,
IL, 60616 USA. \protect
E-mail: \{qian, kai, kren\}@ece.iit.edu.
}}

%
%

\markboth{Journal of \LaTeX\ Class Files,~Vol.~6, No.~1, January~2007}%
{Shell \MakeLowercase{\textit{et al.}}: Bare Demo of IEEEtran.cls for Journals}
%



\maketitle
\thispagestyle{empty}
\pagestyle{empty}

\begin{abstract}
By exploiting multipath fading channels as a source of common randomness, physical layer (PHY) based key generation protocols allow two terminals with correlated observations to generate secret keys with information-theoretical security.
The state of the art, however, still suffers from major limitations,\textit{ e.g.}, low key generation rate, lower entropy of key bits and a high reliance on node mobility. In this paper, a novel cooperative key generation protocol is developed to facilitate high-rate key generation in narrowband fading channels, where two keying nodes extract the phase randomness of the fading channel with the aid of relay node(s).
For the first time, we explicitly consider the effect of estimation methods on the extraction of secret key bits from the underlying fading channels and focus on a popular statistical method--maximum likelihood estimation (MLE).
The performance of the cooperative key generation scheme is extensively evaluated theoretically.
We successfully establish both a theoretical upper bound on the maximum secret key rate from mutual information of correlated random sources and a more practical upper bound from Cramer-Rao bound (CRB) in estimation theory. Numerical examples and simulation studies are also presented to demonstrate the performance of the cooperative key generation system. The results show that the key rate can be improved by a couple
of orders of magnitude compared to the existing approaches.

 \end{abstract}


%
\IEEEpeerreviewmaketitle

\section{Introduction}

\IEEEPARstart{A}{} fundamental problem of all wireless communications is the secure distribution of secret keys, which must be generated and shared between authorized  parties prior to the start of communication. In the field of cryptography, the Diffie-Hellman key exchange protocol is one of the most basic and widely used cryptographic protocols for secure key establishment. The essential idea behind the Diffie-Hellman key exchange is that: two parties that have no prior knowledge of each other to jointly establish a shared secret key over an insecure communication channel. However, the protocol assumes the adversary has bounded computation power  and relies upon computational hardness of
certain mathematical problems to achieve secure key generation. This body of cryptographic protocols achieve \textit{computational security}.


Recently, the notion of  physical layer (PHY) based key generation
has been proposed and the resulting approaches serve as alternative solutions to the key
establishment problem in wireless networks. Based on the theory of reciprocity of antennas and electromagnetic propagation, the channel responses between two transceivers can be used as a source of common randomness that is not available to adversaries in other locations.
Such source of secrecy, which is provided by the fading process of wireless channels, can help to achieve \textit{information-theoretical security}.
This body of work can be traced back to the original information-theoretical formulation of secure communication due to \cite{Sh49}.
Building on information theory and following~\cite{Sh49}, information theorists characterized the fundamental
bounds and showed the feasibility of generating secrets using auxiliary random sources~\cite{Ma97,MaWo03,AhCs93}. However, they are almost all based on theoretical results and do not present explicit constructions. To the best of our knowledge,
Hershey \textit{et al.} proposed the first key generation scheme
based on differential phase detection in \cite{HaSt96}.
Using multipath channels as the source of common randomness, recent researches focus on measuring a popular statistic of wireless channel,
\textit{i.e.}, received signal strength (RSS), for extracting shared secret bits between node pairs
\cite{SaKi07,MaTr08,JaPr09}. It has been demonstrated that these RSS based
methods are feasible on customized 802.11 platforms.
The state of the art, however, still suffers from major limitations. First,
the key bit generation rate supported by these approaches is very
low. This is due to the fact that the PHY based key generation relies on channel variations or node
mobility to extract high entropy bits. In the time intervals where channel changes slowly, only a limited number of key bits can be extracted.
The resulting low key rate significantly limits their practical application given
the intermittent connectivity in mobile environments.
To increase the key rate, Zeng \textit{et al.} proposed a key generation protocol by exploiting multi-antenna diversity~\cite{ZeWu10}. But it also leads to an increase in the complexity of the transceivers.
Second, the generated raw key bit stream has low randomness. This is because the distribution of the RSS measurements or estimates is not uniform, which results in unequally likely bits after quantization. As cryptographic keys need to be as random as possible so that it is infeasible to reproduce them or predict them, it is important to ensure high entropy of the generated keys.
However, the problem of how to safely and efficiently generate random key bits using channel randomness is still open.




To overcome the above limitations, in this paper, we investigate the problem of cooperative key generation between two nodes with the aid of third parties, \textit{i.e.}, relay nodes.
The introduction of the relay nodes is motivated by the \textit{diversity gain} provided by the relay nodes, which can potentially help to increase the key rate by furnishing the two nodes additional correlated randomness. To enhance the level of entropy of bit sequences, we propose to exploit the uniformly distributed channel phase for key generation.
Specifically, we develop a novel time-slotted cooperative key generation scheme by exploiting channel phase randomness under narrowband fading channels.
For the first time,  we explicitly consider the effect of estimation methods on the extraction of secret key bits from the underlying fading channels and focus on a popular statistical method--maximum likelihood estimation (MLE). The main features of the proposed scheme are: i) The  key bit generation rate is improved by a couple
of orders of magnitude compared to RSS based approaches. This is due to the high-accuracy MLE and the fact that the random channels between the relay and the keying nodes can be effectively utilized during a single \textit{coherence time}. That also implies the proposed scheme can even work in a static environment where channels change very slowly; ii) The generated bit stream is very close to a truly
random sequence due to the use of uniformly distributed channel phase for bit generation;
iii) It is robust to relay node compromise attacks since each relay node only contributes a small portion of key bits and a small number of them can
never obtain the complete global key bit information even collectively.
The performance of the cooperative key generation scheme is extensively evaluated theoretically.
We successfully establish both a theoretical upper bound on the maximum secret key rate from mutual information of correlated random sources and a more practical upper bound from Cramer-Rao bound (CRB) in estimation theory.
We also show that the \textit{cooperative gain} in the key generation is similar to the beamforming gain in cooperative networking, \textit{i.e.}, the resulting gain is linear to the number of relay nodes.
Numerical examples and simulation studies are also presented to demonstrate the performance of the cooperative key generation system. The results show that the key rate can be improved by a couple
of orders of magnitude compared to the existing approaches.

The rest of the paper is organized as follows: Section II gives problem formulation and introduces wireless fading channel model considered in this paper. Section III
discusses related work.
Section IV provides the detailed description of our proposed cooperative key generation schemes.
Section V and VI present the theoretical performance analysis and simulation studies, respectively.
Section VII provides a security discussion of the proposed scheme from both practical and theoretical aspects.
Finally, Section VIII concludes the paper.

\section{Problem Formulation and Preliminaries}
In this section, we first define the PHY based key generation problem in wireless networks and introduce the general assumptions made in the existing work~\cite{SaKi07,MaTr08,JaPr09, YeMa10}.
 This will explain why wireless channel between a transmitter-receiver pair can be used as a source of common randomness for secret generation.
Then we discuss two most common channel models and focus on the narrowband fading channel, which is closely related to the key generation schemes developed in this paper.

\subsection{Problem Formulation}

In a multipath fading wireless environment, the physical signals transmitted between a transmitter-receiver pair rapidly decorrelate in space, time and frequency. That implies that it is very hard for a third party to predict the channel state between the legitimate parties, \textit{i.e.}, an eavesdropper at a third location (\textit{e.g.}, one half of wavelength away) cannot observe the same channel response information. This channel \textit{uniqueness} property of the transmitter-receiver pair offers potential  security guarantees. Further, the channel \textit{reciprocity} indicates the availability of using common randomness between the transmitter-receiver pair:
the electromagnetic
waves traveling in both directions will undergo the same
physical perturbations. That implies that in a time-division duplex (TDD) system, if the transmitter-receiver pair
operates on the same
frequency in both directions, the channel states/channel impulse responses
observed at two ends will theoretically be the same. Based on these two observations, we can see that there exists a natural random source in wireless communications for secrecy extraction.

Consider two parties $A$ and $B$ (we term them as \textit{keying nodes} in the following discussion) that want to establish a symmetrical key between them in the presence of an eavesdropper $E$. The keying nodes are assumed to be half-duplex in the sense that they cannot transmit and receive signals at the same frequency simultaneously.
In the first timeslot, $A$ transmits a signal $x_A$ to $B$, and $E$ can also hear this signal over the wireless channel. The signals received by $B$ and $E$ are:
\begin{eqnarray*}
  r_B &=& h_{AB}x_A + n_B \\
  r_E &=& h_{AE}x_A + n_E,
\end{eqnarray*}
where $h_{AB}$ and $h_{AE}$ are the channel gains from $A$ to $B$ and $A$ to $E$, respectively, and $n_B$ and $n_E$ are noises at $B$ and $E$, respectively.
In the second timeslot, $B$ transmits a signal $x_B$ to $A$, and $E$ can also hear this signal over the wireless channel. The signals received by $A$ and $E$ are:
\begin{eqnarray*}
  r_A &=& h_{BA}x_B + n_A \\
  r_E &=& h_{BE}x_B + n_E,
\end{eqnarray*}
where $h_{BA}$ and $h_{BE}$ are the channel gains from $B$ to $A$ and $B$ to $E$, respectively, and $n_A$ and $n_E$ are noises at $A$ and $E$, respectively.
The channel from node $i$ to node $j$ is modeled
as a  multipath fading model with channel impulse $h_{i,j}(t)$.
 We assume channel reciprocity in the forward and reverse directions during the \emph{coherence time} such that  $h_{i,j}(t)=h_{j,i}(t)$
and the underlying noise in each channel is additive white
Gaussian noise (AWGN). In wireless communications, \emph{coherence time} is a statistical measure of the time duration over which the channel impulse response is essentially invariant, and quantifies the similarity of the channel response at different times.

The keying nodes $A$ and $B$ compute the sufficient statistic $\hat{r}_B$ and $\hat{r}_A$, respectively, and generate the secret key based on these observations. In our system, we assume there exist $N$ relay nodes, which are honest and will help and cooperate with the keying nodes $A$ and $B$ to generate secret keys.
On the other side, the eavesdropper $E$ knows the  whole key generation protocol  and can eavesdrop all the communications between legitimate nodes (\textit{i.e.}, $A$, $B$ and relay nodes).
Based on communication theory \cite{Go05}, the signals transmitted between $A$ and $B$ and the signals transmitted between $A$ ($B$) and $E$, which is at least $\lambda/2$ away from the network nodes, experience independent fading.
 As an example, consider a wireless system with
900MHz carrier frequency.
If an eavesdropper $E$ is more than 16cm away from the communicating nodes, it experiences independent channel variations such that no useful information is revealed to it.
Following the same assumptions in most key generation schemes~\cite{SaKi07,JaPr09,WiTs07,YeMa10}, we assume that the adversary $E$ aims to derive the secret key generated between legitimate nodes and further steal the transmitted private information in the future.
 Those active attacks where the attacker tampers the transmissions are orthogonal to our research and thus not considered in this paper.


The above problem can be considered as a key generation problem in cooperative wireless networks in the presence of an eavesdropper. In this paper, we propose to develop an efficient and secure cooperative key generation protocol and provide an information-theoretic study on maximum key rate using techniques from both information theory and estimation theory.
The proposed design should satisfy the following requirements: i) High key rate. Given the intermittent connectivity in mobile environments, the key generation scheme should have a high key rate; ii) Sound key randomness. As cryptographic keys need to be as random as possible so that it is infeasible to reproduce them or predict them, the resulting key bits should have a high level of entropy. Note that the existing schemes usually rely on channel variations or node mobility to extract high entropy bits. We propose to remove this constraint and establish random keys even in static environments.

%
%
%

\subsection{Narrowband and Wideband Fading Channels}\label{Narrowband and Wideband Channels}
An important characteristic of a multipath channel is the \textit{delay spread} $\nu$ it causes to the signal~\cite{Go05}.
If $\nu$ is large, the multipath components are typically resolvable, leading to the wideband fading channel, where the resulting probability distributions for the gains of multipath channel paths are often modeled as log-normal or Nakagami~\cite{WiTs07}.
If $\nu$ is small, the multipath components are typically nonresolvable, leading to the narrowband fading channel, where the amplitude gain is Rayleigh distributed.

In this paper, we will focus on a narrowband fading system for secret key generation. Our approach can also apply to wideband fading channels. But as will be shown, it suits best for narrowband fading channel model. Let the transmitted signal be
\begin{eqnarray*}
  x(t) &=& \mathfrak{R}\{\tilde{u}(t)e^{j2\pi f_c t}\},
\end{eqnarray*}
where $\tilde{u}(t)$ is
the complex envelope of $x(t)$ with bandwidth $B$ and $f_c$ is its carrier frequency.
Assume the equivalent lowpass time-varying channel impulse response is $h(\tau,t)=\sum_{n=0}^{N(t)}\alpha_n(t) e^{-j\phi_n (t)}\delta(\tau-\tau_n(t))$,
the received signal can be written as
%

\begin{eqnarray}
\label{receive signal}
  r(t) &=&  x(t)\ast h(\tau,t)  \\
  &=& \mathfrak{R}\left \{ 	\left (\int\limits_{-\infty}^{\infty}h(\tau,t)\tilde{u}(t-\tau)\, dx  \right )e^{j2\pi f_c t}\right \}\notag \\
    &=& \mathfrak{R}\left \{\left (\sum_{n=0}^{N(t)} \alpha_n(t) e^{-j\phi_n (t)}u(t-\tau_n(t))\right )e^{j2\pi f_c t}  \right \}, \notag
\end{eqnarray}
where $\alpha_n(t)$ is a function of path loss and shadowing while $\phi_n(t)$ depends on delay, Doppler, and carrier
offset. Typically, it is assumed that these two random processes $\alpha_n(t)$ and $\phi_n(t)$ are independent.  Note $N(t)$ is the number of resolvable multipath components. For narrowband fading channels, each term in the sum of Eq.~(\ref{receive signal}) results from nonresolvable multipath components.

Under most delay spread characterizations, $\nu\ll 1/B$ implies that the delay associated with the $k$th multipath
component $\tau_k\leq \nu$ $\forall k$, so $u(t-\tau_k) \thickapprox u(t)$. If $x(t)$ is assumed to be an unmodulated carrier (single-tone signal) $x(t)=\mathfrak{R}\{e^{j2\pi f_c t}\}=\cos2\pi f_c t$, it is narrowband for
\textit{any} $\nu$. With these assumptions, the received signal  becomes
%
%

\begin{eqnarray}
   \label{IQ expression}
 r(t) &=& \mathfrak{R}\left \{ \left (\sum_{n=0}^{N(t)} \alpha_n(t) e^{-j\phi_n (t)}\right) e^{j2\pi f_c t} \right \} \\
   &=&r_I(t)\cos2\pi f_ct-r_Q(t)\sin2\pi f_ct, \notag
\end{eqnarray}
where the in-phase and quadrature components are given by
$ r_I(t)= \sum_{n=1}^{N(t)} \alpha_n(t) \cos \phi_n(t)$ and $r_Q(t)= \sum_{n=1}^{N(t)} \alpha_n(t)$
$\sin \phi_n(t)$, respectively. The in-phase and quadrature components of Rayleigh fading process are jointly Gaussian random process.
The complex ``lowpass'' equivalent signal for $r(t)$ is given by $r_I(t)+jr_Q(t)$ which has phase $\theta=\arctan(r_Q(t)/r_I(t))$, where $\theta$ is uniformly distributed, \textit{i.e.}, $\theta \in \mathcal{U}[0,2\pi]$. So $r_I(t)+jr_Q(t)$ can be written as $r_I(t)+jr_Q(t)=|h|e^{j\theta}=|h|\cos \theta+j|h|\sin \theta$, where $|h|=\sqrt{r_I(t)^2+r_Q(t)^2}$.  Consider the additive white Gaussian noise (AWGN) in the channel, Eq.~(\ref{IQ expression}) can be written as
\begin{eqnarray}
  r(t) &=& |h|\cos \theta \cos 2\pi f_ct - |h|\sin \theta \sin 2\pi f_ct+n(t) \\
       &=& |h|\cos(2\pi f_ct+\theta) + n(t), \notag
       \label{single tone}
\end{eqnarray}
where $n(t)$ is a Gaussian noise process with power spectral density $\frac{N_0}{2}$.
We will estimate parameters in $r(t)$ and use the uniformly distributed phase of multipath channel for key generation. A list of important notation is shown in Table.~\ref{notation}.

\begin{table}[t!]
\caption{A summary of important notation.} 
\centering 
\begin{tabular}{c|l} 
\hline\hline 
Symbol&\multicolumn{1}{c}{Definition}~~~~~~~~~~~~~~ \\ [0.5ex]
\hline 
$p_e$ & the  bit error probability (BER)\\ 
$T_c$ & coherence time\\
$\nu$ & delay spread\\
$q$ & the number of quantization intervals\\
$T_o$ ($T_i$) & observation time or beacon duration time\\
$N_s$ & the number of samples in the observation time\\
$N_0$ & the one-sided power spectra density (PSD) \\
$h_{ji}^I, h_{ij}^I$ & channel gains\\
$f_s$ & sampling rate\\
$N$ & number of relay nodes\\
$R_{k}^{MI}$ & key rate from mutual information with no relay\\
$R_{k}^{CRB}$ & key rate from CRB with no relay\\
$R_{co}^{MI}$ & cooperative key rate from mutual information\\
$R_{co}^{CRB}$ & cooperative key rate from CRB\\[1ex]
\hline 
\end{tabular}
\label{notation}
\end{table}

\section{Related Work}

The PHY based key generation can be traced back to the original information-theoretic formulation of secure communication due to \cite{Sh49}.
Building on information theory, \cite{Ma97,MaWo03,AhCs93} characterized the fundamental
bounds and showed the feasibility of generating keys using external random source-channel impulse response.
To the best of our knowledge, the first key generation scheme suitable for wireless network was proposed in \cite{HaSt96}. In \cite{HaSt96}, the
differential phase between two frequency tones is encoded for key generation. Error control coding techniques are used for enhancing the reliability of key generation.
Similar to \cite{HaSt96}, a technique of using random phase for extracting secret keys in an OFDM system through channel estimation and quantization was recently proposed in \cite{SaPe08}.
This paper characterized the probability of generating the same bit vector between two nodes as a function of signal-to-interference-and-noise (SINR) and quantization levels.

A  key generation scheme based on  extracting secret bits from  correlated deep fades was proposed in \cite{SaKi07} and distinguished from
the aforementioned work by using received signal strength (RSS) as the random source via a TDD link for the protocol design. Two cryptographic
tools-- information reconciliation and privacy amplification are used to eliminate bit vector discrepancies due to RSS measurement asymmetry. The final key agreement is achieved by leaking out minimal information for error correcting and sacrificing a certain amount of entropy for generating
nearly perfect random secret bits.
In \cite{MaTr08}, the authors proposed two key generation schemes based on channel impulse response (CIR) estimation and RSS measurements.
 Different from \cite{SaKi07}, the two transceivers alternately send known probe signals to each another and estimate the magnitude of channel response at successive time instants.
 The excursions in the fading channels are used for generating bits and the timing of excursions  are used for key reconciliation. The resulting sequence are further filtered and quantized using a 1-bit quantizer, which results in low key bit rate. Motivated by observations from quantizing jointly Gaussian process,  a more general key generation scheme was proposed by exploiting empirical measurements to set quantization boundaries in \cite{YeMa10}.
 Working on the same RSS based approach, \cite{JaPr09} evaluated the effectiveness of RSS based key extraction in real
environments. It has been shown that due to lack of channel variations static environments are not suitable for establishing secure keys, and node mobility helps to generate key bits with high entropy. The most recent work \cite{WaSu11} proposed an efficient and scalable key generation scheme that supports both pairwise and group key establishments.

Due to noise, interference and other factors in the key generation process, discrepancies may exist between the generated bit streams.
Variants of this problem have been extensively explored under the names information reconciliation, privacy amplification and fuzzy extractors.
\cite{ReWo04} proposed the first protocol to solve the information-theoretic key agreement problem between two parties that initially posses only correlated weak secrets.
The key agreement was shown to be theoretically feasible when the information that the two bit strings contain about each other is more than the information that the eavesdropper has about them.
\cite{DoKa06} used error-correcting techniques to design a protocol that is computationally efficient for different distance metrics. Based on the previous results, \cite{KaRe09} proposed a protocol that is efficient for both parties and has both lower round complexity and lower entropy loss. Recently, \cite{DoWi09} proposed a two round key agreement protocol for the same settings as \cite{KaRe09}.

\section{The Proposed Solutions}
In this section, we present our cooperative key generation algorithms  for extracting secret bits from  wireless channels. The proposed algorithms employ the technique of \textit{single-tone parameter estimation} to estimate the uniformly distributed channel phase. When keying nodes $A$ and $B$ alternately transmit known single-tone signals to each other, each relay node also observes the fading signals transmitted through the pairwise links between him and the keying nodes. Therefore, with the aid of relay nodes, the keying nodes $A$ and $B$ can potentially increase the key rate using additional randomness in the same \textit{coherence time} interval.

\begin{figure*}[!t]
\begin{center}
\includegraphics [width=15.0cm]{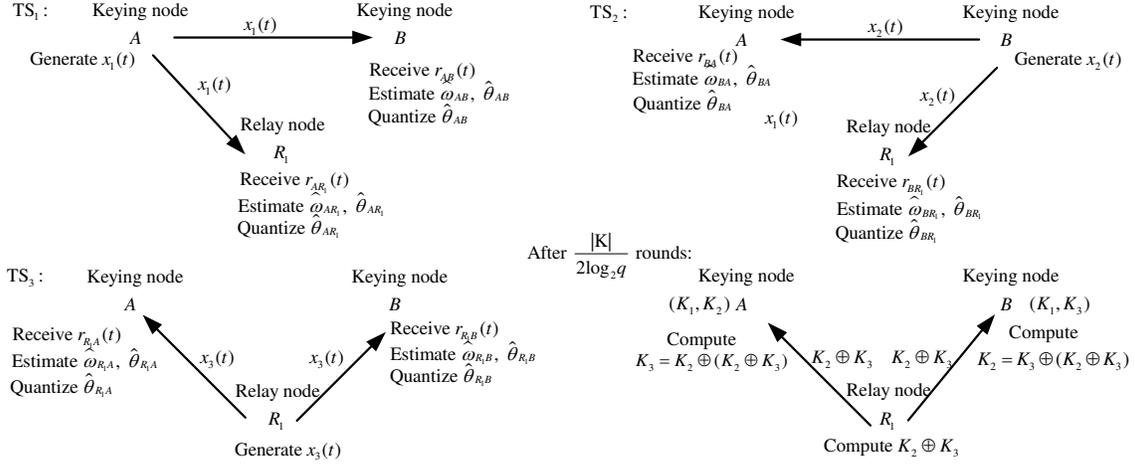}
\end{center}
\caption{Protocol for cooperative key generation with one relay.}
\label{one relay}
\end{figure*}

\subsection{Utilizing a Single Relay}
We fist consider the single relay case where one relay node acts as a helper to facilitate the key generation between the keying nodes $A$ and $B$. The basic idea is that an unmodulated carrier (\textit{i.e.}, single-tone signal) is transmitted through the fading channels back and forth between the keying nodes, and the keying nodes perform maximum Likelihood Estimation (MLE) based on their observation.  Since each bidirectional channel between a pair of nodes is a time-division-duplex (TDD) channel, which is reciprocal in both directions, it will incur the same total phase shift caused by multipath due to the channel \textit{reciprocity principle}.
Generally, the protocol consists of two main phases: i) Single-tone phase estimation and quantization; ii)
Key reconciliation and privacy amplification.

Before we introduce the cooperative key generation protocol, we first introduce the fundamental building block-- MLE used in single-tone signal parameter estimations.
During the protocol execution, the keying nodes $A$, $B$ and relay nodes use MLE to
estimate the parameters of a single-tone signal with a known signal model.
 Given certain observation set $\overline{Z}$ and parameter set $\overline{\alpha}$, the objective of MLE is to estimate the parameter set that maximizes the pdf of $\overline{Z}$. In our application, the received signal model can be written as
 \begin{eqnarray}
   r(t) &=& b_0 \cos (\omega_0t+\theta_0)+n(t),
 \end{eqnarray}
 where $\overline{\alpha}=\{ b_0,\omega_0, \theta_0 \}$ are the unknown parameters (amplitude, frequency and phase, respectively) to be estimated. The received signal is sampled at a constant sampling frequency rate  $f_s=1/T_s$ to produce the discrete-time observation
  \begin{eqnarray}
    r[m]&=&b_0 \cos (w_0(t_0+mT_s)+\theta_0)+n[m]
    \label{discrete signal}
  \end{eqnarray}
for $m=0, 1, \ldots,N_s-1$. Here, $t_0$ denotes the time of the first sample and $n[m]$s are Gaussian random samples with zero mean and variance $\sigma^2$. Let $\overline{Z} =(r[0], r[1], \ldots, r[N_s-1])$,
the pdf of $\overline{Z}$ is~\cite{Rife73}
 \begin{eqnarray*}
f(\overline{Z};\overline{\alpha})&=& \left (\frac{1}{\sigma\sqrt{2\pi}} \right )^{N_s}\exp \left \{-\frac{1}{2\sigma^2}\sum_{m=0}^{N_s-1}(r[m]-\mu[m])^2 \right\},
  \end{eqnarray*}
where $\mu[m]=b_0\cos(w_0(t_0+mT_s)+\theta_0)$.
In the following discussion, we ignore discussion on the estimation of signal amplitude $b_0$ since its estimation is independent from the estimation of frequency and phase.

The $N_s$ samples in Eq. (\ref{discrete signal}) is provided as an input of the MLE estimator.
According to the results in \cite{Rife73}, the maximum of function $f(\overline{Z},\overline{\alpha})$ is achieved when
\begin{equation}\label{phase_est}
    \theta_0=-\tan^{-1}\frac{\sum_{m=0}^{N_s-1} r[m]\sin(\omega m)}{\sum_{m=0}^{N_s-1} r[m] \cos (\omega m)}.
\end{equation}
Thus, we can first estimate the frequency of the signal, and then calculate the ML estimate of the phase using Eq. (\ref{phase_est}). Specifically, the MLE is implemented in three steps:
\begin{enumerate}
  \item \textit{Rough frequency search.} We calculate the Discrete-time Fourier Transformation (DFT) of $\overline{Z}$ and find the $\hat{k}=\arg \max_k |R[\omega_k]|$, where $\omega_k=\frac{2k\pi}{NT_s}$ and $N$ is the length of the DFT. Here, $N$ is chosen to be a power of 2 and greater than $N_s$. Then we can calculate the roughly estimated frequency as $\omega_l=\frac{2\hat{k}\pi}{NT}$. Such frequency estimate has large estimation error due to the limited resolution of the DFT. Thus, a more accurate estimation is desired;
  \item \textit{Fine frequency search.} Based on the rough estimation in the last step, we can calculate the $\hat{\omega}$ by maximizing function $|R(\omega)|$, where $R(\omega)$ is the continuous DFT of the sample sequence $r[m]$ in the interval $[\frac{2(\hat{k}-1)\pi}{N_sT}, \frac{2(\hat{k}+1)\pi}{NT_s}]$. The fine search algorithm locates the value of $\omega$ closest to $\omega_l$ that maximizes $|R(\omega)|$. The \textit{secant method} is used to compute successive approximations
to the frequency estimate $\hat{\omega}=\arg \max_\omega |R(\omega)|$.
  \item \textit{Phase estimation.} The phase estimate can be calculated by substituting $\hat{\omega}$ to Eq. ($\ref{phase_est}$).
\end{enumerate}
The performance of MLE is measured by the variance of the estimation errors. This variance can be lower-bounded by the Cramer-Rao bound (CRB) \cite{RiBo74}. The performance of the ML estimator, which is closely related to the performance of the proposed cooperative key generation scheme, will be discussed and analyzed later.
We present the cooperative key generation protocol as follows (See Fig.~\ref{one relay}):

\noindent \textbf{Phase One: Single-tone phase estimation and quantization}.

\noindent $\textbf{\textup{TS}}_1$:
The protocol begins in timeslot 1  with transmission of a sinusoidal primary beacon of duration $T_1$ from node $A$:
\begin{eqnarray*}
~~~~~~  x_1(t)= a_1\cos(w_c(t-t_1)),
\end{eqnarray*}
where $t\in [t_1,t_1+T_1)$.
To simplify the exposition, we assume $t_1=0$ in the following discussion, \textit{i.e.}, the protocol starts at time zero point.

Node $B$ ($R_1$) observes the initial transient response of the multipath channel  $h_{A,B}(t)$ ($h_{A,R_1}(t)$) to the beacon $x_1(t)$ over the interval
$t\in [\tau_{AB}, \tau_{AB}+\nu_{AB})$ ($t\in [\tau_{AR_1}, \tau_{AR_1}+\nu_{AR_1})$), where $\tau_{AB}$ ($\tau_{AR_1}$)
denotes the delay of the shortest path and $\nu_{AB}$ ($\nu_{AR_1}$) denotes the finite \textit{delay spread} of the channel $h_{A,B}(t)$ ($h_{A,R_1}(t)$).
In order to achieve a  steady-state response  at both $B$ and $R_1$, it
is required that $T_1>\max\{\nu_{AB},\nu_{AR_1}\}$. The ``steady-state'' portion of the beacons
received at $B$ and $R_1$ can be written as
\begin{eqnarray*}
\mathrm{At}~B:~~  r_{AB}(t) &=&  a_1\alpha_{AB}\cos(w_ct+\theta_{AB})+ n_{AB}(t),\\
\mathrm{At}~R_1:~~  r_{AR_1}(t) &=&  a_1\alpha_{AR_1}\cos(w_ct+\theta_{AR_1})+ n_{AR_1}(t),
\end{eqnarray*}
where $~t\in [\tau_{AB}+\nu_{AB}, \tau_{AB}+T_1)$ ($~t\in [\tau_{AR_1}+\nu_{AR_1}, \tau_{AR_1}+T_1)$) for $B$ ($R_1$), and $n_{AB}(t)$ ($n_{AR_1}(t)$) denotes the additive white Gaussian noise (AWGN) in the $A\rightarrow B$ ($A\rightarrow R_1$) channel. $\alpha_{AB}$ ($\alpha_{AR_1}$) and $\theta_{AB}$ ($\theta_{AR_1}$) are the steady-state gain and the phase response of channel $h_{A,B}(t)$ ($h_{A,R_1}(t)$), respectively.
At the end of primary beacon, a  final transient response of the multipath channel is also received by $B$ ($R_1$)
over the interval $t \in [\tau_{AB}+T_1, \tau_{AB}+\nu_{AB}+T_1)$ ($t \in [\tau_{AR_1}+T_1, \tau_{AR_1}+\nu_{AR_1}+T_1)$). $B$ ($R_1$) uses only the steady-state portion
of the noisy observation to compute ML estimates of the received frequency and phase, which are denoted by $\hat{w}_{AB}$ ($\hat{w}_{AR_1}$) and
$\hat{\theta}_{AB}$ ($\hat{\theta}_{AR_1}$), respectively.


\vspace{2mm}
\noindent $\textbf{\textup{TS}}_2$: Upon the conclusion of the primary beacon $r_{AB}(t)$, in timeslot 2,  $B$ begins the transmission of a sinusoidal secondary beacon at
$t_2=\max\{\tau_{AB}+\nu_{AB}+T_1,\tau_{AR_1}+\nu_{AR_1}+T_1\}$. The secondary beacon transmitted by $B$ at $t_2$ can be written as
\begin{eqnarray*}
~~~~~~  x_2(t)= a_2\cos (w_{c}(t-t_2)),
\end{eqnarray*}
where $t\in [t_2,t_2+T_2)$.
$A$ ($R_1$) observes the initial transient response of the multipath channel  $h_{B,A}(t)$ ($h_{B,R_1}(t)$) to beacon $x_2(t)$ over the interval
$t\in [t_2+\tau_{BA}, t_2+\tau_{BA}+\nu_{BA})$ ($t\in [t_2+\tau_{BR_1}, t_2+\tau_{BR_1}+\nu_{BR_1})$), where $\nu_{BA}=\nu_{AB}$ ($\nu_{BR_1}=\nu_{R_1B}$) due to channel reciprocity. In order to achieve a steady-state response at both $A$ and $R_1$,
 $T_2>\max\{\nu_{BA},\nu_{BR_1}\}$ is required. The steady-state portion of the beacons
received at $B$ and $R_1$ can be written as
\begin{eqnarray*}
\mathrm{At}~A:~~  r_{BA}(t) &=&  a_2\alpha_{BA}\cos(w_{c}t+\theta_{BA})+ n_{BA}(t),\\
\mathrm{At}~R_1:~~  r_{BR_1}(t) &=&  a_2\alpha_{BR_1}\cos(w_{c}t+\theta_{BR_1})+ n_{BR_1}(t),
\end{eqnarray*}
 where $t\in [t_2+ \tau_{BA}  +\nu_{BA}, t_2+\tau_{BA}+T_2)$ ($t\in [t_2+ \tau_{BR_1}  +\nu_{BR_1}, t_2+\tau_{BR_1}+T_2)$) for $A$ ($R_1$), and $n_{BA}(t)$ ($n_{BR_1}(t)$) denotes the additive white Gaussian noise (AWGN) in the $B\rightarrow A$ ($B\rightarrow R_1$) channel.
$\alpha_{BA}$ ($\alpha_{BR_1}$) and $\theta_{BA}$ ($\theta_{BR_1}$) are the steady-state gain and the phase response of channel $h_{B,A}(t)$ ($h_{B,R_1}(t)$), respectively.
At the end of this beacon, a final transient response of the multipath channel is received by $A$ ($R_1$)
over the interval $t \in [t_2+\tau_{BA}+T_2, t_2+\tau_{BA}+T_2+\nu_{BA})$ ($t \in [t_2+\tau_{BR_1}+T_2, t_2+\tau_{BR_1}+T_2+\nu_{BR_1})$). Similar to  $\textup{TS}_{1}$, $A$ ($R_1$) uses only the steady-state portion of the noisy observation to compute ML estimates of the received frequency and phase, which are denoted by $\hat{w}_{BA}$ ($\hat{w}_{BR_1}$) and
$\hat{\theta}_{BA}$ ($\hat{\theta}_{BR_1}$), respectively.

\vspace{2mm}
\noindent $\textbf{\textup{TS}}_3$: Upon the conclusion of the primary beacon $r_{BR_1}(t)$, in timeslot 3 $R_1$ begins the transmission of a sinusoidal secondary beacon at
$t_3=\max\{t_2+\tau_{BA}+\nu_{BA}+T_2,t_2+\tau_{BR_1}+\nu_{BR_1}+T_2\}$. The third beacon transmitted by $R_1$ at $t_3$ can be written as
\begin{eqnarray*}
~~~~~~  x_3(t)= a_3\cos (w_{c}(t-t_3)),
\end{eqnarray*}
where $t\in [t_3,t_3+ T_3)$.
$A$ ($B$) observes the initial transient response of the multipath channel  $h_{R_1,A}(t)$ ($h_{R_1,B}(t)$) to beacon $x_3(t)$ over the interval
$t\in [t_3+\tau_{R_1A}, t_3+\tau_{R_1A}+\nu_{R_1A})$ ($t\in [t_3+\tau_{R_1B}, t_2+\tau_{R_1B}+\nu_{R_1B})$), where $\nu_{R_1A}=\nu_{AR_1}$ ($\nu_{R_1B}=\nu_{BR_1}$) due to channel reciprocity. In order to achieve a steady-state response at both $A$ and $B$,
 $T_3>\max\{\nu_{R_1A},\nu_{R_1B}\}$ is required. The steady-state portion of the beacons
received at $A$ and $B$ can be written as
\begin{eqnarray*}
\mathrm{At}~A:~~  r_{R_1A}(t) &=&  a_3\alpha_{R_1A}\cos(w_{c}t+\theta_{R_1A})+ n_{R_1A}(t),\\
\mathrm{At}~B:~~  r_{R_1B}(t) &=&  a_3\alpha_{R_1B}\cos(w_{c}t+\theta_{R_1B})+ n_{R_1B}(t),
\end{eqnarray*}
 where $t\in [t_3+ \tau_{R_1A}  +\nu_{R_1A}, t_3+\tau_{R_1A}+T_3)$ ($t\in [t_3+ \tau_{R_1B}  +\nu_{R_1B}, t_3+\tau_{R_1B}+T_3)$) for $A$ ($B$), and $n_{R_1A}(t)$ ($n_{R_1B}(t)$) denotes the additive white Gaussian noise (AWGN) in the $R_1\rightarrow A$ ($R_1\rightarrow B$) channel.
$\alpha_{R_1A}$ ($\alpha_{R_1B}$) and $\theta_{R_1A}$ ($\theta_{R_1B}$) are the steady-state gain and the phase response of channel $h_{R_1,A}(t)$ ($h_{R_1,B}(t)$), respectively.
At the end of this beacon, a final transient response of the multipath channel is received by $A$ ($B$)
over the interval $t \in [t_3+\tau_{R_1A}+T_3, t_3+\tau_{R_1A}+T_3+\nu_{R_1A})$ ($t \in [t_3+\tau_{R_1B}+T_3, t_3+\tau_{R_1B}+T_3+\nu_{R_1B})$). Similar to  $\textup{TS}_{2}$, $A$ ($B$) uses only the steady-state portion of the noisy observation to compute ML estimates of the received frequency and phase, which are denoted by $\hat{w}_{R_1A}$ ($\hat{w}_{R_1B}$) and
$\hat{\theta}_{R_1A}$ ($\hat{\theta}_{R_1B}$), respectively.

\textit{Quantization.} To generate high-entropy bits, we assume $A$, $B$ and $R_1$ run the above steps once during each \textit{coherence time} interval. For ease of exposition, we term the above steps as round 1.
After round 1, each of the three nodes has two phase estimates for quantization
\begin{eqnarray*}
A:    \hat{\theta}_{BA}~\textbf{mod}~ 2\pi~,~ \hat{\theta}_{R_1A} ~\textbf{mod}~ 2\pi\\
B:     \hat{\theta}_{AB}~\textbf{mod}~ 2\pi~,~ \hat{\theta}_{R_1B} ~\textbf{mod}~ 2\pi\\
R_1:     \hat{\theta}_{AR_1}~\textbf{mod}~ 2\pi~,~ \hat{\theta}_{BR_1} ~\textbf{mod}~ 2\pi
\end{eqnarray*}
Each node uniformly maps their phase estimates into the quantization interval/index using the following formula:
\begin{eqnarray*}
  Q(x) &=& k ~~~\mathrm{if}~x \in [\frac{2\pi(k-1)}{q},\frac{2\pi k}{q})
\end{eqnarray*}
for $k=1,2,\ldots,q$.
Therefore, in the first round, the quantization of
each phase value generates $\log_2(q)$ secret bits. Due to channel reciprocity principle,  $A$ and $B$ share $\log_2(q)$ bits generated from $\hat{\theta}_{BA}$ ($\hat{\theta}_{AB}$); $A$ and $R_1$ share $\log_2(q)$ bits generated from $\hat{\theta}_{R_1A}$ ($\hat{\theta}_{AR_1}$); $B$ and $R_1$ share $\log_2(q)$ bits generated from $\hat{\theta}_{R_1B}$ ($\hat{\theta}_{BR_1}$). Note the quantization index $k$ is encoded into bit vectors. In our implementation, we use \textit{gray codes} to reduce the bit error probability (BER).

Assume the desired key size is $|K|$. For round $k=2, 3, \ldots, \frac{|K|}{2\log_2(q)}$, $A$, $B$ and $R_1$ repeat the operations as in $\textup{TS}_1$, $\textup{TS}_2$ and $\textup{TS}_3$ to generate phase estimates and convert them into bit vectors through $q$-level quantization.

After $\frac{|K|}{2\log_2(q)}$ rounds, a key of size $\frac{|K|}{2}$ is shared between $A$ and $B$, which is denoted as $K_1$. Similarly, a key of size $\frac{|K|}{2}$ is shared between $A$ and $R_1$, which is denoted as $K_2$; a key of size $\frac{|K|}{2}$ is shared between $B$ and $R_1$, which is denoted as $K_3$. Then $R_1$ computes $K_2\oplus K_3$ and transmits it over the public channel. $A$ receives the XOR information and computes $K_2\oplus(K_2\oplus K_3)=K_3$. Similarly, $B$ obtains $K_2$ by $K_3\oplus(K_2\oplus K_3)=K_2$. Now both $A$ and $B$ have keys $K_1, K_2$ and $K_3$.

Finally, $A$ and $B$ set the final key as $K_1||K_2$ or $K_1||K_3$, and a secret key with size $|K|$ is established. Note that we use either $K_2$ or $K_3$ instead of both as the component of the final key. The reason is that with either one of $K_2$ and $K_3$ the eavesdropper can recover the other one by leveraging $K_2\oplus K_3$.

%
%
%




\vspace{2mm}
\noindent \textbf{Phase Two: Key reconciliation and privacy amplification}.
Due to reciprocity principle, the generated bit sequence at $A$ and $B$ should be identical. However, there may exist a small number of bit discrepancies due to estimation errors, hardware variations and half-duplex beacon transmission.
 These error bits can be corrected using key reconciliation techniques~\cite{KaRe09,DoOs08}. Assume $A$ and $B$ hold $K$ and $K'$, respectively. And the Hamming distance
 $\textsf{dis}(K,K')\leq t$.
Following Code-offset construction in \cite{DoOs08}, we use a $[n,k,2t+1]_{2}$ error-correcting code $C$
 to correct errors in $K'$ even though $K'$ may not be in $C$.
 When performing key reconciliation, node $A$ randomly selects a codeword $c$ from $C$ and
 computes \emph{secure sketch} $\textsf{SS}(K)=s=K\oplus c$. Then $s$ is sent to node $B$. Upon receiving $s$,  node $B$ subtracts the shift
 $s$ from $K'$ and gets $\textsf{Rec}(K',s)=c'=K'\oplus s$. Then node $B$ decodes $c'$ to get $c$, and computes $K$ by shifting
 back to get $K=c\oplus s$. Note that since the error-correcting information $s$ is public to both the communicating nodes and the adversary, it can be used by the adversary
to guess portions of the generated key~\cite{JaPr09}.
 To cope with this problem, $A$ and $B$ can further run privacy amplification protocols~\cite{KaRe09} to recover the entropy loss.

%
%

\subsection{Exploiting Multiple Relays}
In this subsection, we present the key generation protocol with multiple relay nodes. As discussed above, when there exists only one relay $R_1$, he can contribute $\log_2q$ bits in each \textit{coherence time} interval. Since the beacon duration (observation time) $T_i$ is relatively small compared to the \textit{coherence time}, a large portion of the \textit{coherence time} interval cannot be effectively utilized. This motivates us to incorporate more relays into the key generation process with potential two advantages: i) the key rate is further increased due to multiple relays' contribution during the same \textit{coherence time} interval. This also implies that even if the nodes or the environment remain static, a key with high entropy can be generated quickly since it employs the randomness of multiple different pairwise links; ii) the security strength is further enhanced as each relay only contributes a small portion of secret bits to the final key. That implies, even if a small number of relays are compromised, the adversary can
never obtain the complete global key bit information.

With the aid of $N$ relay nodes,
the protocol has a total of $N+2$ timeslots for each round (during one \textit{coherence time} interval $T_c$).
Assume the coherence time are roughly divided to $N+2$ portions, each with length $\frac{T_c}{N+2}$. The activities in each timeslot of round 1 are as follows (for ease of exposition, we ignore the explicit value of $t_i$ for $i=1,2,\ldots,N+2$):
\begin{enumerate}
  \item In $\textbf{\textup{TS}}_1$, node $A$ transmits a  sinusoidal primary beacon $x_1(t)$. Node $B$ ($R_j$, $j=\{1, 2, \ldots, N\}$) neglects the initial and final transient portions of the received signal and uses the steady portion to compute the channel phase estimates $\hat{\theta}_{AB}$ ($\hat{\theta}_{AR_k}$).
  \item In $\textbf{\textup{TS}}_2$, node $B$ transmits a sinusoidal secondary beacon $x_2(t)$. Node $A$ ($R_j$, $j=\{1, 2, \ldots, N\}$) neglects the initial and final transient portions of the received signal and uses the steady portion to compute the channel phase estimates $\hat{\theta}_{BA}$ ($\hat{\theta}_{BR_j}$).
  \item In $\textbf{\textup{TS}}_i$ ($i=\{3, 4, \ldots, N+2\}$), node $R_k$ ($j=\{1, 2, \ldots, N\}$) alternately transmits a  sinusoidal beacon $x_{i}(t)$. Nodes $A$ and $B$ neglect the initial and final transient portions of the received signal and use the steady portion to compute the channel phase estimates $\hat{\theta}_{R_jA}$ ($\hat{\theta}_{R_jB}$) for $j=\{1, 2, \ldots, N\}$.
\end{enumerate}

Assume the desired key size is $|K|$. For round $k=2, 3, \ldots, \frac{|K|}{(N+1)\log_2(q)}$, $A$, $B$ and $R_1$ repeat the operations as in $\textup{TS}_1, \textup{TS}_2, \ldots, \textup{TS}_{N+2}$ to generate phase estimates and convert them into bit vectors through $q$-level quantization.

After $\frac{|K|}{(N+1)\log_2(q)}$ rounds, a key of size $\frac{|K|}{N+1}$ is shared between $A$ and $B$, which is denoted as $K_1$. Similarly, a key of size $\frac{|K|}{N+1}$ is shared between $A$ and $R_j$ $(j=\{1, 2, \ldots, N\})$, which is denoted as $K_{j1}$; a key of size $\frac{|K|}{N+1}$ is shared between $B$ and $R_j$ $(j=\{1, 2, \ldots, N\})$, which is denoted as $K_{j2}$. Then $R_j$ computes $K_{j1}\oplus K_{j2}$ and transmits it over the public channel. $A$ receives the XOR information and computes $K_{j1}\oplus(K_{j1}\oplus K_{j2})=K_{j2}$. Similarly, $B$ obtains $K_{j1}$ by $K_{j2}\oplus(K_{j1}\oplus K_{j2})=K_{j1}$. Now both $A$ and $B$ have $2N+1$ keys $K_1, K_{j1}$ and $K_{j2}$ for $j=\{1, 2, \ldots, N\}$.

Finally, $A$ and $B$ set the final key as $K_1||(K_{11}~\mathrm{or}~K_{12} )||(K_{21}~\mathrm{or}~K_{22} )||\cdots||(K_{N1}~\mathrm{or}~K_{N2} )$. The \textit{key reconciliation and privacy amplification} phase is the same as the single relay case. Note that since a single \textit{coherence time} interval is evenly allocated to the keying nodes and relay nodes, the increase of $N$ results in the decrease of available observation time $T_o$ (beacon duration $T_i$). As will be shown later, this would lead to the increase of estimation errors in MLE. Therefore, there must exist an optimal maximum $N$ under which key rate is maximized.

%

\section{Theoretical Performance Analysis}
In this section, we analyze the performance of the cooperative key generation protocol in terms of the maximum key rate the system can achieve. In information theory, the mutual information of two random variables/sequences is a quantity that measures the mutual dependence of the two variables/sequences. Therefore, the secret key rate can be upper bounded by the mutual information between the observations of two transceivers. Motivated by this, we first provide an information-theoretic study into the upper bound on the key rate using mutual information. This bound denotes the maximum key rate that can be generated from the common randomness between the keying nodes. In estimation theory, Cramer-Rao bound provides a lower bound on the variance of biased and unbiased estimators of a deterministic parameter. Since we utilize maximum likelihood estimation (MLE) in our proposed key generation protocol, we also propose to derive a tighter bound on the key rate using the Cramer-Rao bound (CRB).

\subsection{Knowing the Limit: The Upper Bound on Key Rate from Mutual Information}
In this subsection, we analyze the mutual information between the observations of two nodes $i$ and $j$ at two ends of a multipath fading channel. We start the analysis from the no-relay case. As shown above, all the received signals can be expressed as Eq. (\ref{single tone}). These single-tone signals can be precisely reconstructed from samples taken at sampling rate greater or equal at Nyquist rate $f_s=\frac{1}{T_s}=2f_c$ (Note in the following analysis, we choose $f_s\gg 2f_c$). The discrete-time observation at nodes $i$ and $j$ are
\begin{eqnarray}
  r_{ij}[m] &=& a\alpha_{ij}\cos(w_c(t_{ij}+mT_s)+\theta_{ij})+n_{ij}[m] \\
  r_{ji}[m] &=& a\alpha_{ji}\cos(w_c(t_{ji}+mT_s)+\theta_{ji})+n_{ji}[m]
  \label{received_discrete_signal}
\end{eqnarray}
for $m=0, 1, \ldots, N_s-1$, where $t_{ij}$ ($t_{ji}$) denotes the time of the first sample.
Note that when there is no relay,
nodes $A$ and $B$ each can generate $N_s$ samples by fully exploiting the \textit{coherence time} interval. That is, if we neglect the transmission delay, delay spread and processing delay, the observation time (\textit{i.e.}, beacon duration) is $T_o\approx\frac{T_c}{2}$. Thus, $N_s=T_o f_s=\frac{T_cf_s}{2}$.

Let $\textbf{R}_{ij}=[r_{ij}[0],r_{ij}[1],\ldots,r_{ij}[N_s-1]]$ and $\textbf{R}_{ji}=[r_{ji}[0],r_{ji}[1],\ldots,r_{ji}[N_s-1]]$ denote the samples obtained at nodes $j$ and $i$, respectively. According to \cite{WiTs07}, $I(r_{ij}(t);r_{ji}(t))=I(\textbf{R}_{ij};\textbf{R}_{ji})$ as $r(t)$ is fully defined by $\textbf{R}$.

%
%
In practice, given a set $\mathbf{X}$ of independent identically distributed data conditioned on an unknown parameter $\theta$, a sufficient statistic is a function $T(\mathbf{X})$ whose value contains all the information needed to compute any estimate of the parameter (e.g. a maximum likelihood estimate (MLE)). For ease of exposition, we rewrite Eq. (\ref{single tone}) here
 \begin{eqnarray*}
  r(t) &=& |h|\cos \theta \cos 2\pi f_ct - |h|\sin \theta \sin 2\pi f_ct+n(t)\notag \\
       &=& |h|\cos(2\pi f_ct+\theta) + n(t).
 \end{eqnarray*}
In MLE estimation, $|h|\cos(2\pi f_ct+\theta) + n(t)$ is sampled to estimate $|h|$ and $\theta$, where the complex expression of  multipath channel is $h=|h|e^{j\theta}$. Once $|h|$ and $\theta$ are obtained, the terms $|h|\cos \theta \cos 2\pi f_ct$ and $|h|\sin \theta \sin 2\pi f_ct$ are both determined. So it  is equivalent to sample and estimate a signal like $|h|\cos \theta \cos 2\pi f_ct$ or $|h|\sin \theta \sin 2\pi f_ct$ to fully determine the fading channel information. The ``equivalent'' received signals at nodes $i$ and $j$ can be written as
\begin{eqnarray*}
  \overline{r}_{ij}[m] &=& a\alpha_{ij}\cos\theta \cos(w_c(t_{ij}+mT_s))+n_{ij}[m] \\
  \overline{r}_{ji}[m] &=& a\alpha_{ji}\cos\theta \cos(w_c(t_{ji}+mT_s))+n_{ji}[m]
\end{eqnarray*}
for $m=0, 1, \ldots, N_s-1$. Because  $r[m]$ is fully defined by $\overline{r}[m]$ and vice versa, the mutual information between  $r_{ij}[m]$ and $r_{ji}[m]$ is the same as that between $\overline{r}_{ij}[m]$ and $\overline{r}_{ji}[m]$, \textit{i.e.}, $I(\textbf{R}_{ij};\textbf{R}_{ji})=I(\overline{\textbf{R}}_{ij};\overline{\textbf{R}}_{ji})$, where $\overline{\textbf{R}}_{ij}$ and $\overline{\textbf{R}}_{ji}$ are the discrete-time sequences of $\overline{r}_{ij}[m]$ and $\overline{r}_{ji}[m]$, respectively.

Now the problem becomes a Gaussion random variable estimation problem, where the in-phase component $r_I(t)=|h|\cos \theta=\alpha\cos \theta$ is to be estimated (in the following, we abuse standard notation by letting $h^I$ denote the in-phase component).
Let $\textbf{S}_i=\textbf{S}_j=[a\cos(w_c(0T_s)), a\cos(w_c(1T_s)),\ldots,a\cos(w_c(mT_s))]$. Both nodes $i$ and $j$ can compute a sufficient statistic $\widehat{\textbf{R}}_{ji}$ and
$\widehat{\textbf{R}}_{ij}$ for $\overline{\textbf{R}}_{ji}$ and $\overline{\textbf{R}}_{ij}$ respectively~\cite{LaPo09}
\begin{eqnarray}
\widehat{\textbf{R}}_{ji} &=&\frac{\textbf{S}_j^T}{||\textbf{S}_j||^2}\overline{\textbf{R}}_{ji}=h_{ji}^I+\frac{\textbf{S}_j^T}{||\textbf{S}_j||^2}\textbf{N}_i\\
\widehat{\textbf{R}}_{ij} &=&\frac{\textbf{S}_i^T}{||\textbf{S}_i||^2}\overline{\textbf{R}}_{ij}=h_{ij}^I+\frac{\textbf{S}_i^T}{||\textbf{S}_i||^2}\textbf{N}_j
\end{eqnarray}
where$||\textbf{S}_j||^2 =\textbf{S}_j^T\cdot \textbf{S}_j$ and $||\textbf{S}_i||^2 =\textbf{S}_i^T\cdot \textbf{S}_i$.

\vspace{5pt}
\begin{theorem}
\label{theorem:mutual_information}
Let $h_{ji}^I, h_{ij}^I \sim  \mathcal{N}(0,\sigma_h^2)$ and $\textbf{N}_i, \textbf{N}_j \sim  \mathcal{N}(0,\sigma^2) $. Based on sufficient statistics $(\widehat{\textbf{R}}_{ji}, \widehat{\textbf{R}}_{ij})$ at two ends, nodes $i$ and $j$ can generate secret key bits at rate
\begin{eqnarray*}
R_k^{MI} &=&\frac{\ln2}{T_c}\log_2(1+\frac{\sigma_{h}^4N_s^2P^2}{\sigma^4+2\sigma^2\sigma_{h}^2N_sP}),
\end{eqnarray*}
where $P$ denotes the transmission power, $N_s$ denotes the number of samples and $T_c$ is the \textit{coherence time}.
\end{theorem}
\begin{proof}
See Appendix \ref{Proof of Theorem 1}.
\end{proof}

\vspace{5pt}
In the above discussions, we focus on two nodes $i$ and $j$ with no relay node. We next analyze the key rate when there are $N$ relay nodes.
If the sampling rate $f_s$ is fixed, the coherence time $T_c$ which contains $2N_s$ samples is divided into $N+2$ shares.
From the nodes $A$ and $B$'s point of view, they each ``sends'' $\frac{2N_s}{N+2}$ samples. Thus, the \textit{cooperative} key generate rate is
\begin{eqnarray}
R_{co}^{MI}=\frac{(N+1)\ln2}{T_c}\log_2[1+\frac{\sigma_{h}^4(\frac{2N_s}{N+2})^2P^2}{\sigma^4+2\sigma^2\sigma_{h}^2(\frac{2N_s}{N+2})P}]
\end{eqnarray}

Although the mutual information between each node pairs decreases due to the reduction of number of samples, the relay nodes help $A$ and $B$ to establish more
key components, this gain becomes more significant when SNR increases or the channel changes very
slowly. We have the following theorem

\begin{theorem}
When there are $N$ relay nodes, the \textit{cooperative gain} is
\begin{eqnarray}
\lim_{P\to\infty}\frac{R_{co}^{MI}}{R_s^{MI}}&=&N+1\\
\lim_{N_s\to\infty}\frac{R_{co}^{MI}}{R_s^{MI}}&=&N+1,
\end{eqnarray}
where $R_s^{MI}=R_k^{MI}$.

\end{theorem}
As we can see, the gain of cooperative key generation is similar to the beamforming gain in cooperative networking, which is linear to the number of relay nodes.

\subsection{A More Practical Bound: The Upper Bound on Key Rate from Cramer-Rao bound (CRB)}

In the last subsection, we derive a theoretical upper bound on key rate from mutual information. This bound serves as a universal bound in the sense that it does not depend on the specific method of estimation, and it is not tight in general. Therefore, we next
compute a more practical and tighter bound on key rate from Cramer-Rao bound (CRB) in estimation theory.

In the existing RSS based key generation methods, the signal envelops are sampled and quantized for the calculation of secret bits.  By using the signal envelop, there exists a trade-off between the reduction of the sensitivity of the system to timing error and the loss of variability in the resulting key~\cite{WiTs07}. Different from that, in this paper, we use the uniformly distributed channel phase for key generation to achieve a high level of entropy.
One of the most important properties of Maximum Likelihood estimators (MLE) is that it attains the Cramer-Rao bound at least asymptotically. Similarly, starting from the no-relay case, we have the following theorem:

\vspace{5pt}
\begin{theorem}
\label{theorem:CRB}
When maximum likelihood estimation (MLE) and uniform quantization are used, the expected key rate is upper-bounded by
\begin{eqnarray*}
R_k^{CRB}=\frac{\textbf{P}_{QIA} \log_2 q}{T_c},
\end{eqnarray*}
where $\textbf{P}_{QIA}$ is the average probability of quantization index agreement.
\end{theorem}
\begin{proof}
See Appendix \ref{Proof of Theorem 2}.
\end{proof}

When there are $N$ relay nodes, the number of samples at each node is $N_s^{co}=\frac{2N_s}{N+2}$. We substitute $N_s$ for $N_s^{co}$ in Eq.~(\ref{CRB_theta}) and obtain the new CRB for $\tilde{\theta}$. This bound is used to calculate $\textbf{P}_{QIA}^{co}$. Thus, the expected key rate for cooperative key generation becomes
\begin{eqnarray}
R_{co}^{CRB}&=&\frac{(N+1)\textbf{P}_{QIA}^{co} \log_2 q}{T_c}.
\end{eqnarray}

\begin{figure}[!t]
\begin{center}
\includegraphics [width=9.0cm]{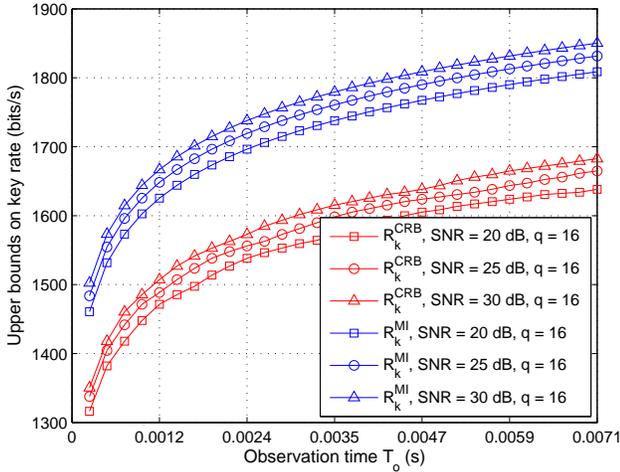}
\end{center}
\caption{Key rate versus observation time $T_o$ under different SNRs.}
\label{observation_timeVSkey_rate}
\end{figure}

It is easy to see that as $q$ increases, node $i$ and $j$ could generate a longer bit vector during the same \textit{coherence time} $T_c$. However, due to estimation errors the probability of generating the same bit vector becomes less.
We can derive the maximum key agreement rate when $q$ satisfies
\begin{eqnarray}
\frac{\partial R_{co}^{CRB}}{\partial q}=0.
\label{maximum_q}
\end{eqnarray}

From the above discussion, we conclude that there exists an optimal $q$ at which maximum key rate can be achieved. We demonstrate how key rate changes as a function of $q$ through simulations in Section~\ref{simulation study}.

\subsection{Numerical Illustration on Theoretical Upper Bounds}
 Assume \textit{coherence time} $T_c=$ 14$ms$.
The example in Fig.~\ref{observation_timeVSkey_rate} presents the two upper bounds on key rate between two nodes (\textit{i.e.}, no relay) as the observation time $T_o$ increases. The results show that the upper bound derived from mutual information serves as the universal upper bound on key rate. As expected, with a fixed number of quantization levels, the increase of SNR or $T_o$ leads to the increase of key rate. Since there are only two nodes, the observation time for each node can be up to 7ms. When $T_o$ changes from 0 to 2.4ms, key rate increases rapidly, and it increases almost linearly as a function of $T_o$ after 2.4ms. Hence, a less observation time can be properly chosen to still maintain an acceptable level of key rate. On the other hand, while the maximum $T_o$ is constrained by $T_c/2$,  one can further enhance the key rate by increasing SNR.

 Fig.~\ref{number_of_relaysVSkey_rate}  plots the upper bounds on key rate when the number of relays $N$ increases. The close match of the bound from mutual information and the bound from CRB before $N=500$ shows that, the CRB can be used to efficiently approach the universal upper bound when the nodes use ML phase estimation. Recall that as $N$ increases, the observation time $T_o$ for each node decreases because the whole \textit{coherence time} are equally distributed to the keying nodes and relay nodes. Due to the fact that the decrease of $T_o$ causes more estimation errors, there exists a threshold on key rate.
This can be clearly observed from the results: the bound based on CRB gradually achieves the maximum and decreases after $N=2500$.  For the sake of clearly illustrating the inflection point on the bound curve from CRB, we limit the range of $N$ in the figure. In fact, there also exists a inflection point on the bound curve from mutual information when $N$ goes to infinity.


\begin{figure}[!t]
\begin{center}
\includegraphics [width=8.5cm]{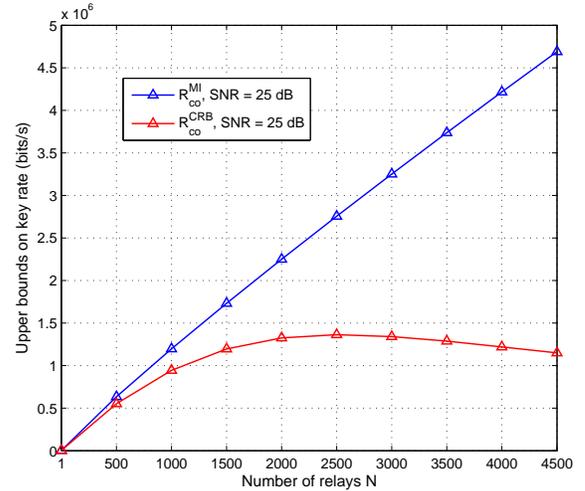}
\end{center}
\caption{key rate versus the number of relays $N$. Note that the observation time $T_o$ is not fixed, \textit{i.e.}, $T_o$ decreases as $N$ increases.}
\label{number_of_relaysVSkey_rate}
\end{figure}

\begin{table}\centering
\caption{Simulation Configuration}
\label{sim_config}
{%
\begin{tabular}{|l|l|}
\hline
Carrier frequency $f_c$   & 900 MHz\\\hline
Sampling frequency $f_s$ & 2.7 GHz\\\hline
Average moving speed $v$     & 10 m/s\\\hline
Coherence time $T_c$  & 14 ms\\\hline
Node distance $d$   & 2 m -- 10 m\\\hline
Delay spread $\nu$    & 1.2$ \mu$s\\\hline
\end{tabular}}
\end{table}%

\emph{\textbf{Discussion}}. In our protocol, the keying nodes rely on a common time reference  to generate \textit{absolute} phase estimates.
If there exists no common time reference among the nodes, each node has to count on its own local time obtained from its local oscillator.
 This implies that the phase estimate generated by each node will has an ``unknown'' offset associated with the node itself,
which prevents the key generation protocol from working correctly. As a future direction, it is worthwhile to extend our protocol to overcome the effect of unknown phase offsets and allow key generation in the unsynchronized case.

We are also going to build a simple prototype to validate the effectiveness of the protocol.
The nodes can be implemented by TMS320C6713 DSKs boards, and the primary beacons can be generated and sent out by a function generator, \textit{e.g.}, HP33120A.
In the implementation, we can use phase-locked loops (PLLs) to realize phase and frequency estimation functions for improving the efficiency.
Since each node transmits a periodic extension of a beacon received in a previous timeslot, the phase and
frequency estimation functions during the synchronization timeslots can be realized by using phase-locked
loops (PLLs) with holdover circuits, \textit{i.e.},  the PLLs are required for each node to store its local phase and frequency estimates during protocol execution.

\section{Simulation Studies}\label{simulation study}

\subsection{Key Rate and Bit Error Probability}
This section presents simulation results of the cooperative key generation protocol in  multipath fading channels.
 In our simulation, we sample the beacon signal with sampling rate $f_s=3f_c$, where $f_c=900$ MHz is the carrier frequency of the single-tone signal.
  In a mobile scenario, we assume the moving speed $v=10$m/s. Thus, the Doppler frequency shift is $f_d=\frac{v}{\lambda}=30$Hz, which results in a \textit{coherence time} $T_c=\frac{0.423}{f_d}=14$ms.
  Assume  $\nu$ is the delay spread with a typical value  $1.2\times10^{-6}$s and the distance  $d$ between nodes changes from 2m to 10m. Thus, the random propagation delay $\tau=\frac{d}{c}=6.67$ns $\sim 33.3$ns.
  We choose $T_o$ much larger than the delay spread $\nu$ so that steady-state response can be achieved. The simulation settings are summarized in Table~\ref{sim_config}.
  Two different methods are used here to
estimate the variance of the phase estimation error: (i) full ML estimation
 and (ii) approximate analytical predictions using CRB.

\begin{figure}[!t]
\begin{center}
\includegraphics [width=9.0cm]{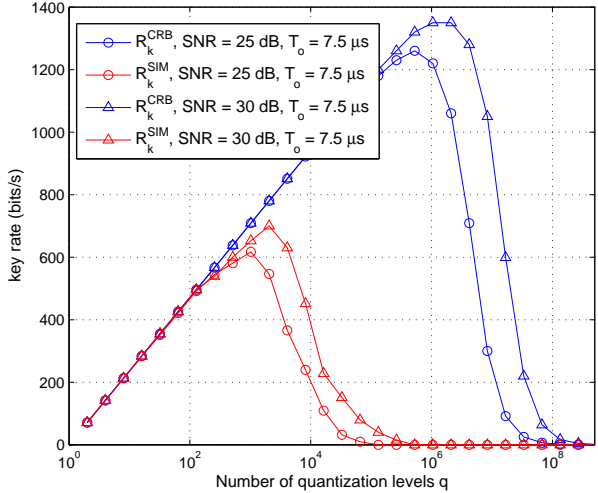}
\end{center}
\caption{Key rate versus the number of quantization levels $q$.}
\label{quantization_levelVSkey_rate}
\end{figure}

\begin{figure}[!t]
\begin{center}
\includegraphics [width=9.0cm]{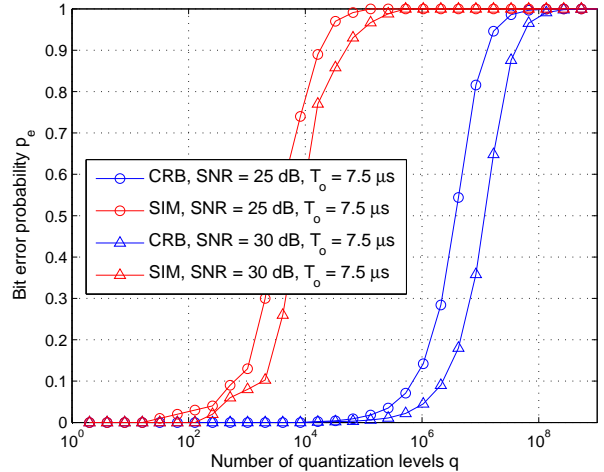}
\end{center}
\caption{Bit probability error $p_e$ versus the number of quantization levels $q$.}
\label{quantization_levelVSPe}
\end{figure}

\begin{figure}[!t]
\begin{center}
\includegraphics [width=9.0cm]{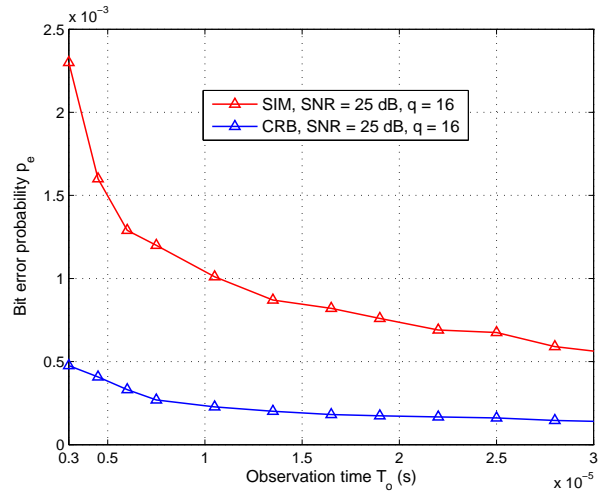}
\end{center}
\caption{Bit error probability $p_e$ versus observation time $T_o$.}
\label{observation_timeVPe}
\end{figure}

\begin{figure}[!t]
\begin{center}
\includegraphics [width=9.0cm]{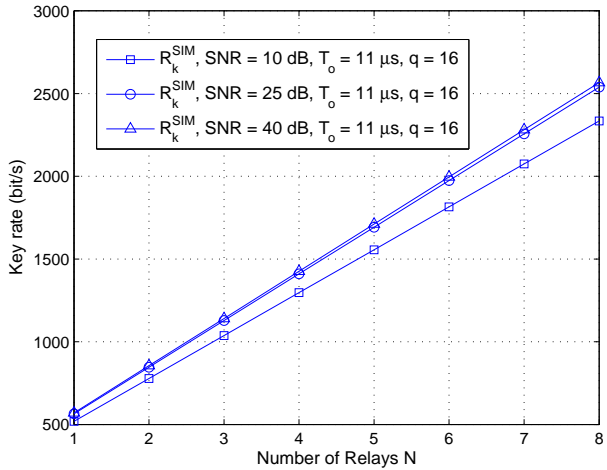}
\end{center}
\caption{Key rate versus the number of relays $N$.}
\label{number_of_relaysVSkey_rate_sim}
\end{figure}

The first example considers the effect of quantization level $q$ on key rate. Fig.~\ref{quantization_levelVSkey_rate} plots the key rate versus $q$ given SNR $=$25dB and $T_o=$7.5$\mu$s using both the CRB analytical predictions and simulations. The results show two regimes of operation. In the small-quantization level regime, the effect of $\log_2q$ dominates the key rate. In this regime, the probability that two estimates fall into the same interval $\textbf{P}_{QIA}$ is very high. Thus, the increase of $q$ leads to the increase of key rate. According to Eq.~(\ref{maximum_q}), when $q$ begins to exceed a threshold, the key rate begins to decrease and enters into the large-quantization level regime. In this regime, the key rate decreases quickly as $q$ further increases. This is due to the fact that the estimation errors dominate the performance as the length of each interval $\frac{2\pi}{q}$ decreases, \textit{i.e.}, $\textbf{P}_{QIA}$ is very sensitive to the estimation errors when the length of interval is small. As might be expected, the CRB can be used to efficiently predict the performance when $q$ is relatively small, \textit{e.g.}, $q<10^3$ in this setting. Since CRB is a lower bound on the variance of the estimation error, it takes a much larger $q$ to reach the inflexion point compared to the simulation results. The above result intuitively suggests that an optimal $q$ can be chosen to maximize the key rate.
To evaluate the BER performance,  Fig.~\ref{quantization_levelVSPe} plots the bit error probability between two nodes as a function of $q$. The results show that, with a fixed $T_o=7.5\mu$s, $p_e$ can be maintained at a very low level if $q<100$. We can use Gray codes (one bit of error is introduced between adjacent sectors) to encode the quantization indices to reduce $p_e$. Also note that in these results, the coherence time is not fully exploited (\textit{i.e.}, the observation time $T_o=7.5 ~\mu$s$\ll T_c$), so one can also reduce $p_e$ so as to increase key rate by setting a larger $T_o$.


Fig.~\ref{observation_timeVPe} plots bit error probability $p_e$ as a function of observation time $T_o$ under SNR$=25$ dB and $q=16$. The results show that the increase of $T_o$ is equivalent to the increase of SNR, which results in a close match of simulation results and CRB.
Fig.~\ref{number_of_relaysVSkey_rate_sim} plots the key rate of the cooperative key generation protocol as the number of relay nodes increases when the quantization levels is fixed at $q=16$.
We choose $T_o=11~\mu$s to maintain a high level of estimation accuracy. The results show that key rate increases linearly as a function of $N$, which confirms our previous analysis that the gain of cooperative key generation scales with the number of relays. As a final point on the results, we note that the further increase of SNR (\textit{e.g.}, from 25 dB to 40 dB) does not help much to improve the performance. This is because the estimation accuracy is already high enough when choosing a short $q$ and a reasonable value of $T_o$.

%
%
%
%
%

\subsection{Key Randomness and The Effect of Mobility}
As we discussed above, the proposed cooperative key generation scheme employs the inherent randomness of uniformly distributed channel phases in multipath narrowband fading channels.
 We employ a widely used randomness test suite NIST to verify the randomness of the secret-bit generated from our simulation~\cite{NIST}. To pass the test, all p-values must be greater than $0.01$.
 In the test, we randomly select 10 bit sequences generated from our simulation and compute their p-values for 8 tests.
 The results in Table~\ref{NIST} show that the average entropy of our generated bit sequences is very close to a truly random sequence.


\begin{table}[!t]
\centering
\begin{tabular}{l c}
\toprule
$\bold{TEST}$ & $\bold{P-value}$ \\
\hline \hline
 & Avg \\
 \hline
DFT & 0.6039 \\
\hline
Lempel Ziv Comp. & 0.4453  \\
\hline
Monobit Freq. & 0.5547  \\
\hline
Runs & 0.4045  \\
\hline
Approximate Entropy & 0.5869 \\
\hline
Cumu. Sums (Forward) & 0.5951 \\
\hline
Cumu. Sums (Reverse) & 0.5887  \\
\hline
Block Frequency & 0.5732  \\
\hline
Serial & 0.5732, 0.5091  \\

\bottomrule
\end{tabular}
\caption{Results of NIST.}
\label{NIST}
\end{table}

%

\section{Security Analysis}
In this section, we provide a security discussion for the proposed cooperative key generation scheme. We focus on both practical and analytical aspects.
The security of the proposed key generation scheme is guaranteed based on the assumption that the adversary is not located near the legitimate parties, \textit{i.e.}, $A$, $B$ and other relay nodes. The is due to the spatial decorrelation fact: since the signal decorrelates over a distance of approximately one half length~\cite{Go05}, it is almost impossible for an adversary which is
located at a different place with the transceivers to obtain the identical channel response for key
generation.
That is, an entity which is at least $\lambda/2$ away from the network nodes experiences fading channels to the nodes are statistically independent of the channels between the communicating nodes. As an example, consider a wireless system with
900MHz carrier frequency.
If the adversary is more than 16cm away from the communicating nodes, it experiences independent channel variations such that no useful information is revealed to it.
By passively observing the signals transmitted between legitimate nodes, it has been empirically shown in~\cite{YeMa10} that the eavesdropper cannot obtain any significant information about the signals received at legitimate nodes.

Another key point regarding the security aspect is that we rely on the uniformity of the channel phase for extracting secret key bits in the narrowband fading channels. As discussed in Section~\ref{Narrowband and Wideband Channels}, the complex lowpass equivalent signal for $r(t)$ can be written as $r_{LP}=r_I(t)+jr_Q(t)$, where the phase of $r(t)$ is $\theta=\arctan(\frac{r_Q(t)}{r_I(t)})$. For uncorrelated Gaussian random variables $r_I(t)$ and $r_Q(t)$, it can be shown that $\theta$ is uniformly distributed over $[0,2\pi]$~\cite{Go05}. Consequently, our proposed PHY based key generation algorithm is best suited for the narrowband fading channels, where $r(t)$ has a Rayleigh-distributed amplitude and uniform phase.
We have the following theorem:



\vspace{5pt}
\begin{theorem}
\label{theorem:security}
 The cooperative key generation scheme is secure, \textit{i.e.}, the resulting secret key is effectively concealed from the eavesdropper observing the public information:
 \begin{eqnarray*}
 \frac{1}{N+1}I(M_0, M_1, M_2, \ldots, M_N;K_{AB}, K_{11}, K_{21}, \ldots, K_{N1})
\leq \epsilon \notag
\end{eqnarray*}
\end{theorem}
\begin{proof}
See Appendix~\ref{Proof of Theorem 3}.

\end{proof}



\section{Conclusion}
In this paper, a novel cooperative key generation protocol was developed to facilitate high-rate key generation in narrowband fading channels, where two keying nodes extract the phase randomness of the fading channel with the aid of relay node(s).
For the first time, we explicitly considered the effect of estimation methods on the extraction of secret key bits from the underlying fading channels and focused on a popular statistical method--maximum likelihood estimation (MLE).
The performance of the cooperative key generation scheme was extensively evaluated theoretically.
We successfully established both a theoretical upper bound on the maximum secret key rate from mutual information of correlated random sources and a more practical upper bound from Cramer-Rao bound (CRB) in estimation theory. Numerical examples and simulation studies were also presented to demonstrate the performance of the cooperative key generation system. The results show that the key rate can be improved by a couple
of orders of magnitude compared to the existing approaches.



%
%

\appendices
\section{Proof of Theorem 1}\label{Proof of Theorem 1}

\begin{proof}
From the above discussion, it is easy to see that $\widehat{\textbf{R}}_{ij} $ is a zero mean Gaussian random variable with variance $\sigma_{h}^2+\frac{\sigma^2}{||\textbf{S}_i||^2}$.
Similarly, $\widehat{\textbf{R}}_{ji}$ is a zero mean Gaussian random variable with variance
$\sigma_{h}^2+\frac{\sigma^2}{||\textbf{S}_j||^2}$. Assume that nodes $i$ and $j$ transmit with power $P=\frac{a^2}{2}$. We have $||\textbf{S}_i||^2=||\textbf{S}_j||^2 \approx PN_s$.
Obviously, $(\widehat{\textbf{R}}_{ij},\widehat{\textbf{R}}_{ji})$ retains all the common randomness in $(\overline{\textbf{R}}_{ij};\overline{\textbf{R}}_{ji})$.
Thus, the mutual information
\begin{eqnarray}
  I(r_{ij}(t);r_{ji}(t)) &=& I(\overline{\textbf{R}}_{ij};\overline{\textbf{R}}_{ji})\\ \notag
                         &=& I(\widehat{\textbf{R}}_{ij};\widehat{\textbf{R}}_{ji}). \notag
\end{eqnarray}
The mutual information $I(\widehat{\textbf{R}}_{ij};\widehat{\textbf{R}}_{ji})$
can be computed as follows
\begin{eqnarray}
I(\widehat{\textbf{R}}_{ij};\widehat{\textbf{R}}_{ji}) =
H(\widehat{\textbf{R}}_{ij})+H(\widehat{\textbf{R}}_{ji})-H(\widehat{\textbf{R}}_{ij},\widehat{\textbf{R}}_{ji})~~~~~~~~~~~~~\\ \notag
=\frac{\ln2}{2}\log_2\left (2\pi e(\sigma_{h}^2+\frac{\sigma^2}{PN_s})\right)~~~~~~~~~~~~~~~~~\\ \notag
~~~~~~~~+\frac{\ln2}{2}\log_2\left (2\pi
e(\sigma_{h}^2+\frac{\sigma^2}{PN_s})\right)-H(\widehat{\textbf{R}}_{ij},\widehat{\textbf{R}}_{ji})\\ \notag
~~~~~~~~~~~~~~~~~~=\ln2\log_2\left (2\pi
e(\sigma_{h}^2+\frac{\sigma^2}{PN_s})\right)-H(\widehat{\textbf{R}}_{ij},\widehat{\textbf{R}}_{ji}).\notag
\end{eqnarray}

Obviously, $\widehat{\textbf{R}}_{ij}$ and $\widehat{\textbf{R}}_{ji}$ form a multivariate normal distribution, thus
\begin{eqnarray}
H(\widehat{\textbf{R}}_{ij},\widehat{\textbf{R}}_{ji})=\frac{\ln2}{2}\log_2[(2\pi e)^2 \textbf{det}(\Sigma)],
\end{eqnarray}
where $\Sigma$ is the covariance matrix of vector $\left[\widehat{\textbf{R}}_{ij},\widehat{\textbf{R}}_{ji} \right]^T$, \emph{i.e.},
\begin{eqnarray}
\Sigma=\begin{bmatrix}\sigma_{h}^2+\frac{\sigma^2}{PN_s} &
\textbf{Cov}(\widehat{\textbf{R}}_{ij},\widehat{\textbf{R}}_{ji})\\ \textbf{Cov}(\widehat{\textbf{R}}_{ij},\widehat{\textbf{R}}_{ji}) &\sigma_{h}^2+\frac{\sigma^2}{PN_s}\end{bmatrix}.
\end{eqnarray}

The covariance of $\widehat{\textbf{R}}_{ij},\widehat{\textbf{R}}_{ji}$ is calculated by
\begin{eqnarray}
\textbf{Cov}(\widehat{\textbf{R}}_{ij},\widehat{\textbf{R}}_{ji})&=&\mathbb{E}(\widehat{\textbf{R}}_{ij}\widehat{\textbf{R}}_{ji})-\mathbb{E}[\widehat{\textbf{R}}_{ij}]\mathbb{E}[\widehat{\textbf{R}}_{ji}]\\ \notag
&=&\mathbb{E}\left [(h_{ij}^I+\frac{\textbf{S}_j^T}{||\textbf{S}_j||^2}\textbf{N}_i)(h_{ji}^I+\frac{\textbf{S}_i^T}{||\textbf{S}_i||^2}\textbf{N}_j)\right]\\ \notag
&=&\mathbb{E}[h_{ij}^2]\\ \notag
&=&\sigma_{h}^2.
\end{eqnarray}
And $\textbf{det}(\Sigma)$ is the determinant of $\Sigma$, which is computed by
\begin{eqnarray}
\textbf{det}(\Sigma) &=&(\sigma_{h}^2+\frac{\sigma^2}{PN_s})^2-\sigma_{h}^4\\ \notag
&=&\frac{2\sigma_{h}^2\sigma^2}{PN_s}+\frac{\sigma^4}{P^2N_s^2}.\notag
\end{eqnarray}

Thus, the mutual information between nodes $i$ and $j$ is
\begin{eqnarray}
I(\widehat{\textbf{R}}_{ij};\widehat{\textbf{R}}_{ji})&=&\ln2\log_2(1+\frac{\sigma_{h}^4N_s^2P^2}{\sigma^4+2\sigma^2\sigma_{h}^2N_sP}).
\end{eqnarray}
Assume the \textit{coherence time} is $T_c$, the maximum key rate is
\begin{eqnarray}
R_k^{MI} &=&\frac{1}{T_c}I(\widehat{\textbf{R}}_{ij};\widehat{\textbf{R}}_{ji})\\ \notag
&=&\frac{\ln2}{T_c}\log_2(1+\frac{\sigma_{h}^4N_s^2P^2}{\sigma^4+2\sigma^2\sigma_{h}^2N_sP}), \notag
\end{eqnarray}
where the superscript $MI$ in $R_k^{MI}$ denotes that the key rate is derived as an upper bound from mutual information.
\end{proof}

\section{Proof of Theorem 2}\label{Proof of Theorem 2}

\begin{proof}
To facilitate analysis, we assume that when the number of samples increases by using larger observation time, the estimation errors converge to
zero-mean Gaussian random variables with variances $\sigma_{\tilde{\theta}}^2$,  which can be lower-bounded by
the Cramer-Rao bounds (CRB)~\cite{RiBo74}.
Fig.~\ref{error_distribution} plots both the distribution of the MLE errors using simulation and the CRB results. The simulation results show that variance of the estimation errors $\sigma_{\sigma_{\tilde{\theta}}^2}^{SIM}=1.6877\cdot 10^{-6}$ is lower-bounded by the CRB $\sigma_{\sigma_{\tilde{\theta}}^2}^{CRB}=1.5616\cdot 10^{-6}$.
 When
estimating the unknown phase of a sampled
sinusoid of amplitude $a$ in white noise with Power Spectral Density (PSD) $\frac{N_0}{2}$, the CRB for the variance of the phase estimate is given as
\begin{eqnarray}
\sigma_{\tilde{\theta}}^2\geq \frac{4f_s\sigma^2(2N_s-1)}{a^2N_s(N_s+1)}\approx \frac{4N_o}{a^2T_o}\approx\frac{4N_0f_s}{a^2N_s},
\label{CRB_theta}
\end{eqnarray}
where $f_s$ is the sampling rate, $N_s$ is the number of samples in
the observation, and $T_o$ is the observation time (\textit{i.e.}, beacon duration) in
second. The approximations can be obtained by assuming that $N_s$ is
large and the fact that $N_s/f_s = T_o=\frac{T_c}{2}$.

\begin{figure}[!t]
\begin{center}
\includegraphics [width=8.0cm]{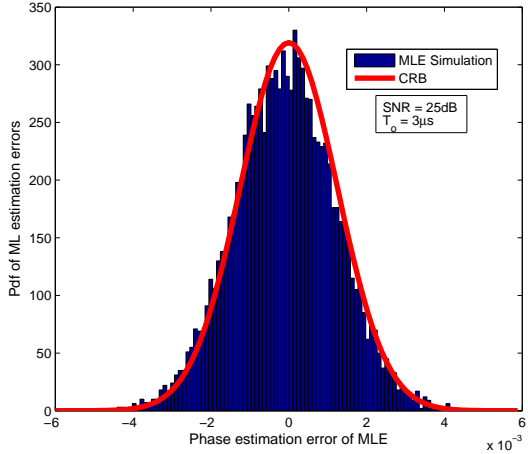}
\end{center}
\caption{The comparison of ML estimation error distribution using simulation and CRB.}
\label{error_distribution}
\end{figure}

Consider Eq. (\ref{received_discrete_signal}), we assume $a_r=a\alpha$ is the received signal strength (we neglect the subscript $i,j$ for simplicity).
The amplitude response of the
fading channel $\alpha$ is Rayleigh distributed, and $\mathbb{E}[\alpha^2]=2\sigma_{h}^2$, then
$a_r^2=2\sigma_{h}^2a^2$. Hence, the CRB bound for the received signal can be expressed as a function of SNR and $N_s$

\begin{eqnarray}
\sigma_{\tilde{\theta}}^2 \geq \frac{4}{ \mathrm{SNR} N_s},
\label{receive_crb}
\end{eqnarray}
where
\begin{eqnarray}
 \mathrm{SNR} =\frac{a_r^2}{2N_0f_s}=\frac{2\sigma_{h}^2 P}{\sigma^2}.
\end{eqnarray}

Suppose $[0,2\pi]$ is divided into $q=2^n$ levels. Now we analyze the probability that nodes $i$ and $j$'s estimations fall into the same interval when performing quantization. Let $\textbf{P}_{QIA}$ denote the average probability of quantization index agreement.
Without loss of generality, assume that $\theta $ falls into  the $i$-th sector $[\frac{2\pi i}{q},\frac{2\pi(i+1)}{q})$ $(i\in\{0,1,\cdots, q-1\})$.
As phase estimation errors are independent and
Gaussian distributed according to the CRB in Eq.(\ref{receive_crb}),
the probability that $\hat{\theta}=\theta+\tilde{\theta} \in [\frac{2\pi i'}{q},\frac{2\pi(i'+1)}{q})$ is (see Fig. \ref{error_distribution_on_inverval})
\begin{eqnarray}
\textbf{P}_{i'}(\theta)&=&\int_{\frac{2\pi i'}{q}}^{\frac{2\pi(i'+1)}{q}}\frac{1}{\sqrt{2\pi}\sigma_{\tilde{\theta}}}e^{-\frac{(x-\theta)^2}{2\sigma_{\tilde{\theta}}^2}} dx,
\label{probability of Pi}
\end{eqnarray}
where $i'\in\{0,1,\cdots, q-1\}$ and $\tilde{\theta}$ is the estimation error.

Thus, $\textbf{P}_{QIA}$ can be computed as $\textbf{P}_{QIA}(\theta)=\sum_{i'=0}^{q-1}\textbf{P}_{i'}(\theta)^2$.
Note that $\textbf{P}_{QIA}(\theta)$ is a function of $\theta$. The value of $\textbf{P}_{QIA}(\theta)$
 goes up when the ``true'' $\theta$ approximates the center of a sector and down when $\theta$ is close to the boundaries of a sector.
In fact, given $\phi\in [0,2\pi]$, $\textbf{P}_{QIA}(\theta)$ is symmetric to the center of a sector and is changing periodically with period $2\pi/q$.
Our simulation results indicate that the variance of phase estimate is much smaller than one.
Thus, given $\theta\in[\frac{2\pi i}{q},\frac{2\pi(i+1)}{q})$, $\textbf{P}_{QIA}(\theta)$ is mainly determined by $\textbf{P}_{i}(\theta)$ ($i'=i$). Based on the above analysis, we can compute the average probability of quantization index  agreement as
\begin{eqnarray}
   \textbf{P}_{QIA} &=& \int_{\frac{2\pi i}{q}}^{\frac{2\pi(i+1)}{q}} \textbf{P}_{QIA}(\theta) \frac{q}{2\pi}d\theta\\
  &\approx& \int_{\frac{2\pi i}{q}}^{\frac{2\pi(i+1)}{q}}\textbf{P}_{i}^2(\theta) \frac{q}{2\pi}d\theta.\notag
\end{eqnarray}

\begin{figure}[!t]
\begin{center}
\includegraphics [width=9.0cm]{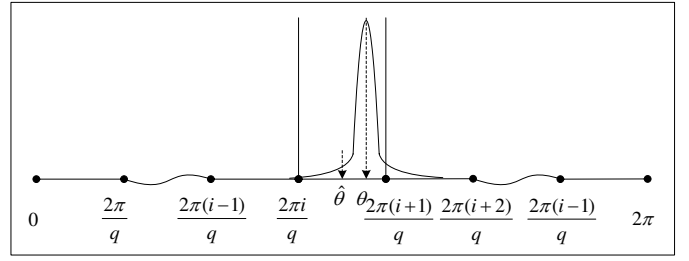}
\end{center}
\caption{An illustration of estimation error distribution on quantization intervals.}
\label{error_distribution_on_inverval}
\end{figure}

%

When nodes $i$ and $j$'s estimates lie in the same interval, they agree on a bit vector of length $\log_2 q$; otherwise they agree on no bit. Hence, the expected key rate is
\begin{eqnarray}
R_k^{CRB}=\frac{\textbf{P}_{QIA} \log_2 q}{T_c}.
\end{eqnarray}

Note that $p_e\approx1-\textbf{P}_{QIA}$ if we assume zero bits are generated when two nodes' estimates fall into different intervals. If gray codes are utilized, $p_e\approx1-\textbf{P}_{QIA}/\log_2q$.
\end{proof}

\section{Proof of Theorem 3}\label{Proof of Theorem 3}
\begin{proof}
Assume $N$ relay nodes are involved with the key establishment. An eavesdropper $E$
monitors all the communications and tries to use these information to find the secret key. Without loss of generality, we assume the key can be established in one round. We have
\begin{eqnarray}
I(\widehat{\textbf{R}}_{AB};\widehat{\textbf{R}}_{BA}) &=& K_{AB}\\
I(\widehat{\textbf{R}}_{AR_j};\widehat{\textbf{R}}_{R_jA}) &=& K_{j1}\\
I(\widehat{\textbf{R}}_{BR_j};\widehat{\textbf{R}}_{R_jB}) &=& K_{j2}
\end{eqnarray}
Suppose $A$ and $B$ always choose $K_{j1}$ as their key component. Let
$M_0=\{\widehat{\textbf{R}}_{AE},\widehat{\textbf{R}}_{BE}\}$. The information $E$ could learn during the
agreement of $K_{j1}$ is $ M_j=\{\widehat{\textbf{R}}_{AE},\widehat{\textbf{R}}_{BE},\widehat{\textbf{R}}_{R_jE},K_{j1}\oplus K_{j2}\}$. Because channels between any two pair of nodes are independent, hence, for any $\epsilon>0$,
we have
\begin{eqnarray}
I(\widehat{\textbf{R}}_{AE},\widehat{\textbf{R}}_{BE};K_{AB})&\leq&\epsilon\\
I(\widehat{\textbf{R}}_{AE},\widehat{\textbf{R}}_{BE},\widehat{\textbf{R}}_{R_jE};K_{j1})&\leq&\epsilon,
\end{eqnarray}
After the relay node $R_j$ broadcasts $K_{j1}\oplus K_{j2}$, $E$ learns $K_{j1}\oplus K_{j2}$. However
\begin{eqnarray}
I(K_{j1}\oplus K_{j2};K_{j1})&=&0.
\end{eqnarray}
It is equivalent to a one-time-pad encryption on $K_{j1}$ with secrete key $K_{j2}$.
Without knowing $K_{j2}$, $E$ could learn nothing from the ciphertext $K_{j1}\oplus K_{j2}$, thus we have
\begin{eqnarray}
I(M_j;K_{j1})=I(\widehat{\textbf{R}}_{AE},\widehat{\textbf{R}}_{BE},\widehat{\textbf{R}}_{R_jE};K_{j1})+\\\notag
I(K_{j1}\oplus K_{j2};K_{j1})\leq \epsilon. \\\notag
\end{eqnarray}

The total information obtained by $E$ is the set $\{M_0, M_1, M_2,\ldots,M_N\}$, whose elements are independent of each other.
On the other side, $A$ and $B$ obtain the key set $\{K_{AB}, K_{11}, K_{21}, \ldots, K_{N1}\}$, whose elements are also independent of each other.
According to the independence of the random variables and the basic properties of mutual information, we have
\begin{eqnarray*}
I(M_0, M_1, M_2, \ldots, M_j;K_{AB}, K_{11}, K_{21}, \ldots, K_{j1})~~~~~~~~~\\
=I(M_0;K_{AB})+\sum_{j=1}^n I(M_j;K_{j1})\leq (N+1)\epsilon~~~~~~~~~\\
\end{eqnarray*}

\end{proof}

\end{document}